\documentclass{article}
\usepackage{graphicx,amsthm,amssymb,amsmath,mathtools,complexity}

\renewcommand{\R}{\mathbb{R}}
\newcommand{\RR}{\mathcal{R}}
\newcommand{\BB}{\mathcal{B}}

\newcommand{\oh}{\mathrm{oh}}
\newcommand{\opt}{\mathrm{opt}}
\newcommand{\holes}{\mathrm{holes}}
\newcommand{\FILL}{\textsc{Fill}}

\newcommand{\mysum}[3]{\;\smashoperator{\sum_{#1}^{#2}}\;{#3}\;}

\newtheorem{theorem}{Theorem}
\newtheorem{lemma}[theorem]{Lemma}
\newtheorem{definition}[theorem]{Definition}
\newtheorem{cor}[theorem]{Corollary}

\usepackage{xcolor}

\title{Approximability of (Simultaneous) Class Cover for Boxes}

\author{Jean Cardinal \thanks{Department of Computer Science, Université libre de Bruxelles, {\tt jcardin@ulb.ac.be}}
\and Justin Dallant \thanks{Department of Computer Science, Université libre de Bruxelles, {\tt justin.dallant@ulb.be} Supported by the French Community of Belgium via the funding of a FRIA grant.}
\and John Iacono \thanks{Department of Computer Science, Université libre de Bruxelles, {\tt john@johniacono.com} Supported by the Fonds de la Recherche Scientifique-FNRS under Grant no MISU F 6001 1.}
}

\date{}

\index{Cardinal, Jean}
\index{Dallant, Justin}
\index{Iacono, John}

\begin{document}
\thispagestyle{empty}
\maketitle

\begin{abstract}
Bereg et al.\ (2012) introduced the Boxes Class Cover problem, which has its roots in classification and clustering applications: Given a set of $n$ points in the plane, each colored red or blue, find the smallest cardinality set of axis-aligned boxes whose union covers the red points without covering any blue point. In this paper we give an alternative proof of $\APX$-hardness for this problem, which also yields an explicit lower bound on its approximability. Our proof also directly applies when restricted to sets of points in general position and to the case where so-called half-strips are considered instead of boxes, which is a new result. 

We also introduce a symmetric variant of this problem, which we call Simultaneous Boxes Class Cover and can be stated as follows: Given a set $S$ of $n$ points in the plane, each colored red or blue, find the smallest cardinality set of axis-aligned boxes which together cover $S$ such that all boxes cover only points of the same color and no box covering a red point intersects a box covering a blue point. We show that this problem is also $\APX$-hard and give a polynomial-time constant-factor approximation algorithm. 
\end{abstract}

\section{Introduction}
Many approaches to data mining, classification and clustering tasks show a close proximity to computational geometry, often embedding data in some space and reducing or formalizing the considered problem as a geometric one. Some well-known approaches include support vector machines, nearest neighbour classifiers or k-means clustering. In this paper we study the computational hardness and approximability of two variants of another geometric problem which has its roots in classification and clustering, known as the Class Cover problem. It can be stated as follows: given a set of $n$ points, each colored red or blue, find the smallest number of balls centered at red points which cover all red points without covering any blue points (see Section 2 for the exact definition of ``cover'' we use here). This problem was introduced by Cannon and Cowen \cite{Cannon2004}, motivated by connections to the measure of separability between two classes of points defined by Cowen and Priebe \cite{Cowen1997} as well as applications to classification and data reduction. In this paper they showed that the problem was $\NP$-hard but admitted a polynomial-time $(1+\ln(n))$-factor approximation. In the Euclidean setting for constant dimension they show that the problem even admits a polynomial-time approximation scheme ($\PTAS$). The problem was studied by others in the context of applications to pattern recognition \cite{Devinney2003} and classifiers \cite{Priebe2003}.

Here we study a variant introduced by Bereg et al.\ \cite{Bereg2012} which can be formulated as follows:
\begin{definition}[Boxes Class Cover (BCC)]
Given a set of $n$ points in the plane, each colored red or blue, find the smallest cardinality set of axis-aligned boxes which together cover the red points without covering any blue point.
\end{definition}
The authors show that this variant is $\NP$-hard and give a polynomial-time algorithm which achieves a $O(1+\log c)$-approximation, where $c$ is the size of an optimal cover. They also study a few restricted cases, among those covering with squares or so-called half-strips, where they show that they remain $\NP$-hard but admit a $O(1)$-factor approximation in polynomial time. The variant for squares was later shown to admit a $\PTAS$ \cite{Aschner2013} while for the general BCC problem it was shown by Shanjani \cite{Shanjani2020} that no $\PTAS$ can exist unless $\P=\NP$ (in this paper the author also notes that the original reduction used by Bereg et al.\ in fact already shows this result).

Note that in all the variations mentioned above, there is a clear asymmetry between the roles of the two color classes, which is not always warranted in applications. As a measure of separability of two classes this can also lead to a phenomenon where the first class can be ``separated'' from the second using few boxes, while separating the second from the first requires many boxes. This motivates us to consider a symmetric version of the problem, which we formulate as follows:
\begin{definition}[Simultaneous BCC (SBCC)]
Given a set $S$ of $n$ points in the plane, each colored red or blue, find the smallest cardinality set of axis-aligned boxes which together cover $S$ such that all boxes cover only points of the same color and no box covering a red point intersects the interior of a box covering a blue point.
\end{definition}

In Section 2 we go over some basic definitions and lemmas. In Section 3 we give an alternative proof to the fact that there is no $\PTAS$ for BCC unless $\P=\NP$, by a reduction from Vertex Cover. While this is already known, our proof is more direct than previous ones and allows us to exhibit a specific lower bound on the approximation factor. The proof also works for half-strips and for the restriction where we consider only point sets in general position, which are new results. In Section 4 we show that for the SBCC problem there is again no $\PTAS$ unless $\P=\NP$. Finally in Section 5 we give a polynomial-time constant-factor approximation algorithm for SBCC. In the way to doing so we show that requiring all boxes to be independent in the SBCC blows up the size of an optimal solution by at most a factor of $9$.

\section{Some definitions and lemmas}

We start by giving some definitions of the object we will consider in this paper. Note that we sometimes use the name of a problem (such as ``BCC'') to denote the corresponding structure we want to minimize (in case of a BCC, this would be a set of axis-aligned boxes covering the red points without covering any blue point).

\begin{definition}
A \emph{bichromatic} set of points is a set of points where some points are said to be red and the rest are said to be blue.
\end{definition}
We will use the notation $S=R\cup B$ to denote such a set of points, where $R$ is the set of red points, $B$ is the set of blue points, and $R\cap B = \emptyset$.
In what follows and in the rest of this paper, we will consider only closed axis-aligned boxes with non-empty interiors and unless otherwise specified, they are bounded. 

We use the following slightly technical definition of a rectilinear polygon, which will make some of the results and proofs easier to phrase:
\begin{definition}
A \emph{rectilinear polygon} $P$ is a connected set of points in the plane such that:
\begin{itemize}
    \item the closure of $P$ can be obtained as the union of a finite number of axis-aligned boxes $\BB$,
    \item every pair of boxes in $\BB$ is either disjoint or intersects in more than one point,
    \item and all connected components of $\R^2\setminus P$ have non-empty interior.
\end{itemize}
\end{definition}

Note that from this definition a rectilinear polygon is not necessarily closed (nor open).
We extend the usual definitions of vertices and edges of a polygon to rectilinear polygons. Note that such a polygon may or may not have holes (bounded connected components of $\R^2\setminus P$) which also contribute vertices and edges. 

We add a few more basic definitions:
\begin{definition}We call \emph{outer-hull} of $P$, denoted as $\oh(P)$, the union of all edges adjacent to the unbounded connected component of $\R^2\setminus P$. The \emph{outer-complexity} of $P$ is the number of vertices in $\oh(P)$. A \emph{convex vertex} of $P$ is a vertex such that the right angle formed by the two edges adjacent to the vertex is directed towards the interior of $P$. Otherwise we call the vertex \emph{reflex}.
\end{definition}
For a polygon $P$, we let $|P|$ denote the number of vertices in $P$. The same applies for a union of polygons or the outer-hull of a (union of) polygon(s). In all other cases we let $|S|$ denote the cardinality of the set $S$ (in such cases $S$ will always be finite).

\begin{definition}
Let $S=R\cup B$ be a bichromatic set of points. We say that an axis-aligned box $K$ \emph{covers a point} $p$ if $p$ is in the interior of $K$. We say that $K$ is red (resp.\ blue) if it covers only red (resp.\ blue) points of $S$. It is \emph{monochromatic} if it is either red or blue.

We say that a set $Z$ of axis-aligned boxes is a \emph{BCC} of $S$ if all red points in $S$ are covered by some box in $Z$ and all boxes in $Z$ are red. We say that $Z$ is an \emph{SBCC} $S$ if all points in $S$ are covered by some box in $Z$, all boxes in $Z$ are monochromatic and no two boxes of different color have their interiors intersecting.

We say that $Z$ \emph{covers a rectilinear polygon} $P$ if $P$ is included in the union of all boxes in $Z$. We say that this cover is \emph{exact} (or that $Z$ covers $P$ exactly) if the closure of $P$ is equal to the union of all boxes in $Z$.
\end{definition}

We finish this section with two lemmas which will prove useful to us.
\begin{lemma}\label{lemma:refl-conv-holes}
Let $P$ be a rectilinear polygon. Let $c$ be the number of convex vertices of $P$, $r$ be the number of reflex vertices and $h$ be the number of holes. Then we have $r=c+4(h-1)$. In particular, if $P$ is obtained as the union of $k$ axis-aligned boxes, then it has outer-complexity less than $8k$, as each box contributes at most $4$ convex vertices to the outer face.
\end{lemma}
\begin{proof}
It is well known (see for example \cite{ORourke1983}) that for a rectilinear polygon with no holes we have $r=c-4$. Now consider some hole $H$ of $P$. If we view $H$ as a rectilinear polygon itself and denote $c_H$ and $r_H$ the number of convex and reflex vertices of $H$ respectively, then we have $r_H=c_H-4$ (as $H$ has no hole). Because a convex (resp.\ reflex) vertex for $H$ is a reflex (resp.\ convex) vertex for $P$, the hole $H$ contributes $r_H+4$ reflex vertices and $r_H$ convex vertices to $P$. As for a rectilinear polygon with no holes we have $r=c-4$ and every hole contributes $4$ more reflex vertices than convex vertices, it holds that in a a rectilinear polygon with $h$ holes we have $r=c-4 + 4h = c+4(h-1)$.
\end{proof}

\begin{lemma}[\cite{Eppstein2009}]\label{lemma:ind-rec-cover}
Any rectilinear polygon with $n$ vertices and $h$ holes can be covered exactly with at most $n/2+h-1$ axis-aligned boxes.
\end{lemma}

\section{An approximation-preserving reduction from Minimum Vertex Cover to Minimum Class Cover}

Here we show how to reduce the (minimum) Vertex Cover problem to BCC in a way that shows the following:
\begin{theorem}\label{thm:class-cover-reduction}
If BCC can be approximated in polynomial time within a constant factor of $(1+\epsilon)$ for $\epsilon>0$ then Vertex Cover can be approximated in polynomial time within a constant factor of $(1+(d+1)\epsilon)$ on graphs with maximum degree bounded by $d$.
\end{theorem}

Using the fact that Vertex Cover is $\NP$-hard to approximate within a constant factor of $1+1/52$ on graphs with maximum degree at most $4$ \cite{Chlebik2003} we also get the following:
\begin{cor}
Approximating BCC within a constant factor of $1+1/260$ ($\approx 1.0038)$ is $\NP$-hard.
\end{cor}

Note that as mentionned in the introduction, the reduction used by Bereg et al.\ from Rectilinear Polygon Cover (which itself admits a reduction from Vertex Cover in bounded degree graphs) also proves $\APX$-hardness. However the reduction here is more direct, which allows us to easily get the more precise $(1+(d+1)\epsilon)$ relationship and an explicit lower bound on the approximation factor. It also has the advantage that it is easy to see that it works for points in general position also (by perturbing the points slightly), whereas this is not immediately the case for the reduction from Rectilinear Polygon Cover.

\begin{figure}
    \centering
    \includegraphics[scale=0.3]{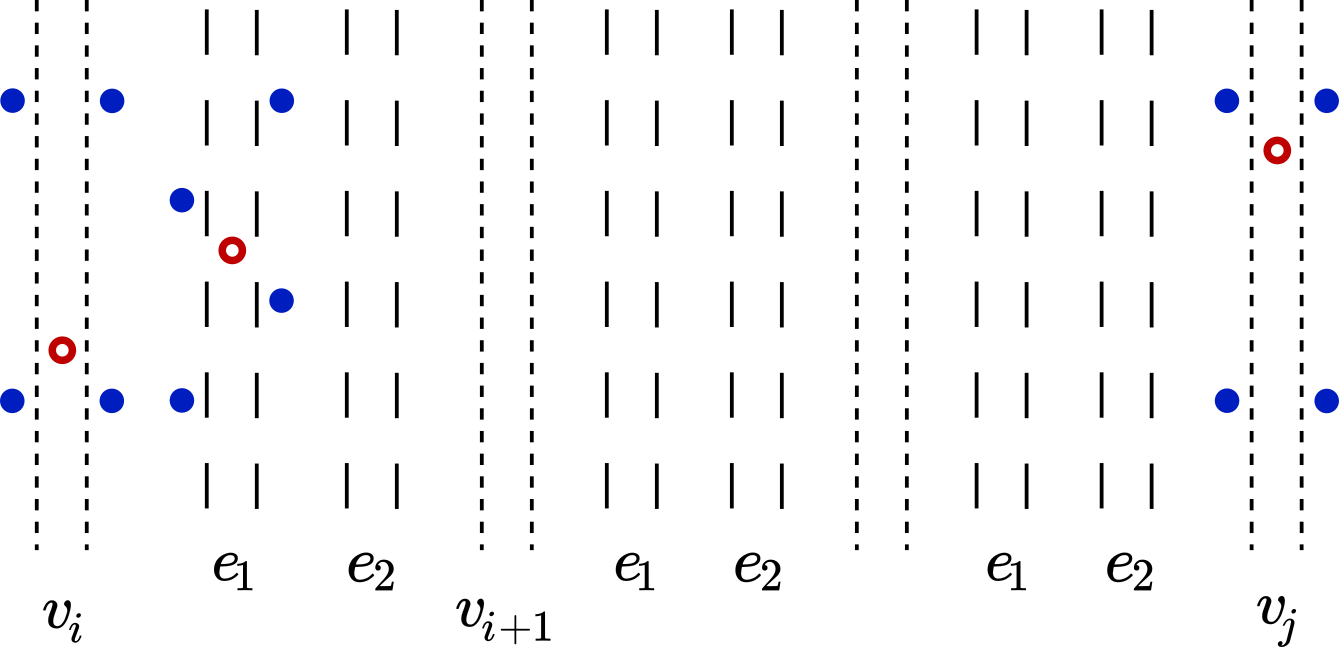}
    \caption{The gadget corresponding to an edge $e_1$ between vertices $v_i$ and $v_j$. The vertex lanes are represented with short dashes and the edge lanes with long ones. Here there are two edge lanes between consecutive vertex lanes, thus representing a graph with two edges (assuming all lanes appear on this figure).}
    \label{fig:edge-gadget}
\end{figure}


\begin{proof}[of Theorem \ref{thm:class-cover-reduction}]
Suppose we are given a simple undirected graph $G$ with $n$ vertices denoted as $v_1,v_2\ldots,v_n$ and $m$ edges denoted as $e_1,e_2\ldots,e_m$, of maximum degree at most $d$.

Imagine creating $n$ disjoint vertical slabs in the plane, one for each vertex (in the order given by the vertex indices). We call such a slab a vertex lane. Between any two consecutive vertex lanes, we create $m$ additional disjoint vertical slabs (also disjoint from the previously created vertex lanes), one for each edge. We call these the edge lanes.

Now we describe a gadget which will encode an edge of $G$ as a set of points in the plane to cover. Consider an edge $e$ between $v_i$ and $v_j$, $i<j$. We place a red point on the vertex lane corresponding to $v_i$ and one with larger $y$-coordinate on the vertex lane corresponding to $v_j$. We also add a red point on the unique edge lane corresponding to $e$ between the lanes of $v_i$ and $v_{i+1}$, with $y$-coordinate between those of the two previously placed points. Next, we add a few blue points which restrict the type of boxes which can be used to cover these red points (see Figure \ref{fig:edge-gadget}). We create an instance $S$ of the BCC problem by placing the edge gadgets such that their minimal bounding boxes are all disjoint. 

Consider the minimum vertex cover of $G$. Denote its size as $\opt_G$. For every vertex in this cover we can create a box covering the corresponding vertex lane entirely. We thus obtain a set of $\opt_G$ boxes such that in every gadget, either the left or right vertex lane is covered (or both). In each gadget we need exactly one additional box to fully cover the red points. Thus, we obtain a solution to the BCC problem of size $\opt_G + m$. In particular, if we let $\opt_{S}$ denote the minimum size of a BCC of $S$, we have $\opt_S \leq \opt_G + m$.

Now consider a solution to the BCC problem using $b$ boxes. We assume without loss of generality that no two boxes cover the exact same subset of red points. 

Consider the gadget corresponding to some edge. Notice that a red box covering the middle red point cannot cover any red point in any different gadget (as this point is the only one in its edge lane). Moreover, a box covering the left or right point can only additionally cover the middle point or other points on the same vertex lane (but not both simultaneously). 

If there is a box covering the middle point which covers neither the left nor right point, we can either delete this box (if another box also covers the middle point) or replace it with a box covering, say, both the middle and left point. This does not make any other red point uncovered. Consider this change done from now on. If the middle red point is covered twice, then we can replace the box which also covers, say, the right red point with a box covering the whole vertex lane corresponding to that right point without making any red point uncovered. If any vertex lane is fully covered multiple times after this process we can simply discard all but one of the boxes covering it.

\begin{figure}
    \centering
    \includegraphics[scale=0.33]{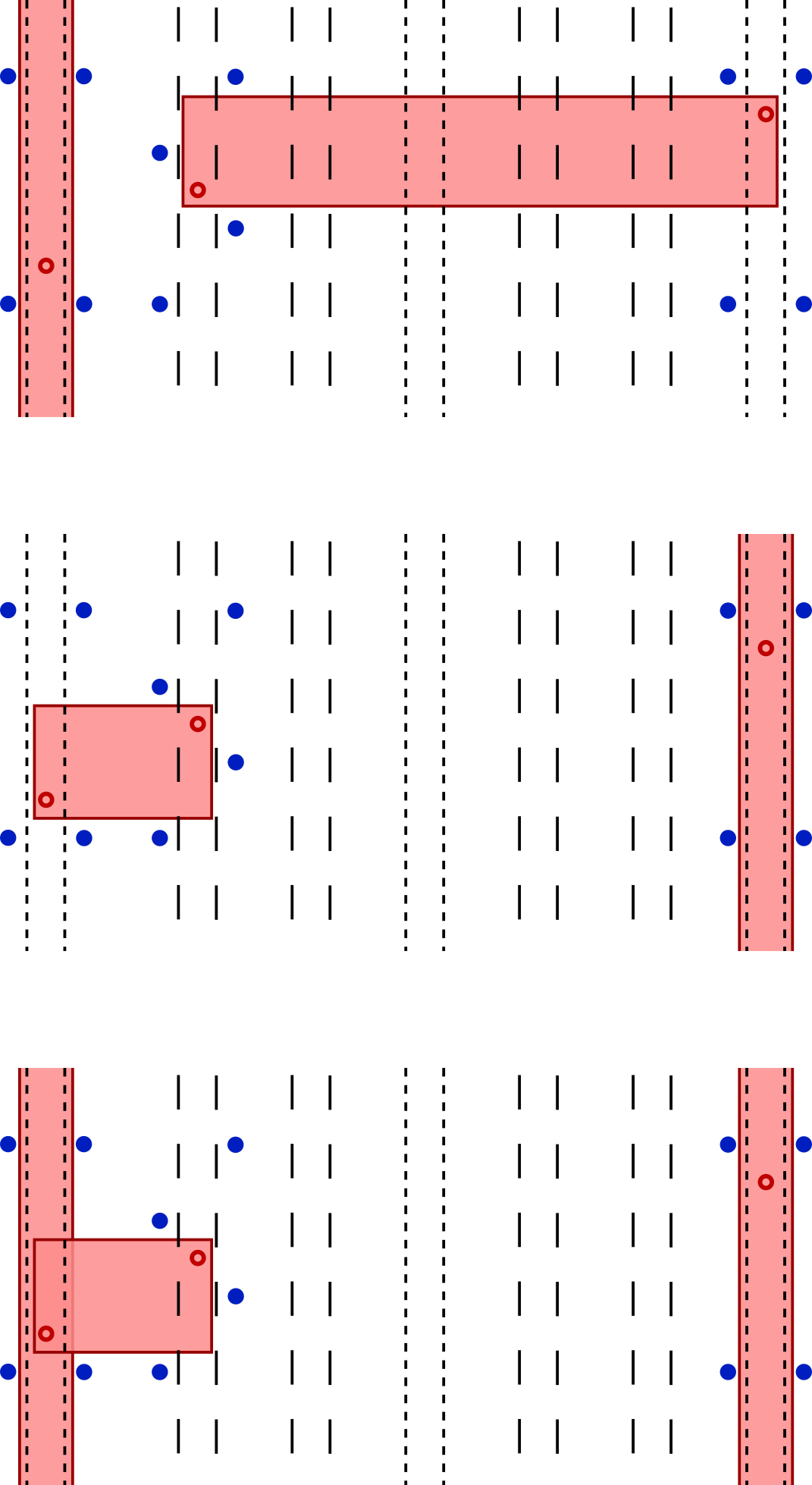}
    \caption{Regardless if the left vertex lane, right vertex lane or both are covered, there is always exactly one additional box covering the remaining red points in an edge-gadget.}
    \label{fig:covered_lanes}
\end{figure}

After making these changes, we end up with a set of at most $b$ boxes such that for any of the $m$ edge-gadgets, either the lane corresponding to the left point is fully covered or the lane corresponding to the right point is fully covered (or both). In all cases, there is one additional box covering the middle point (see Figure \ref{fig:covered_lanes}. Thus, if $\ell$ is the number of lanes covered, the total number of boxes is $\ell + m \leq b$. We can create a vertex cover of $G$ of size $\ell \leq b-m$ by choosing every vertex such that the corresponding lane is fully covered.

Now say that we can approximate BCC in polynomial time within a factor of $(1+\epsilon)$. We can run this algorithm on $S$ to obtain a class cover of size at most $(1+\epsilon)\opt_S$. By the process described above we can then obtain a vertex cover of $G$ of size at most
\begin{align*}
    (1+\epsilon)\opt_S-m &\leq (1+\epsilon)(\opt_G + m)-m \\
    &=(1+\epsilon)\opt_G + \epsilon\cdot m.
\end{align*}

Because $G$ is of maximum degree at most $d$, every vertex can cover at most $d$ edges and we have $\opt_G \geq \frac{m}{d}$, i.e.\ $m\leq d\cdot \opt_G$. Thus, the vertex cover we obtain is of size at most
\begin{align*}
    (1+\epsilon)\opt_G + \epsilon\cdot m 
    &\leq (1+\epsilon)\opt_G + \epsilon\cdot d\cdot \opt_G \\
    &= (1+(d+1)\epsilon)\opt_G.
\end{align*}

This concludes the proof.
\end{proof}

Note that this proof works exactly the same if we replace the axis-aligned boxes with axis-aligned half-strips (axis-aligned boxes which are unbounded in one of the four axis-aligned directions). By a simple perturbation argument, it also yields the same results when restricted to sets of points in general position (that is, where no three points are collinear). 

\section{An approximation hardness proof for Simultaneous Boxes Class Cover}

In this section we prove that SBCC is $\APX$-hard. Here, the proof strategy used in the previous section breaks down because of the interactions between the boxes covering the red points and those covering the blue points. Instead, we revert back to the proof strategy used by Bereg et al.\, by a reduction from Rectilinear Polygon Cover. 

\begin{definition}[Rectilinear Polygon Cover (RPC)]
Given a closed rectilinear polygon $P$, find the smallest cardinality set of axis-aligned boxes covering $P$ exactly.
\end{definition}

In our case we are covering the red and blue points simultaneously so we consider a slightly different problem, where we want to cover $P$ and its complement simultaneously.

\begin{definition}[Simultaneous RPC (SRPC)]
Given a closed rectilinear polygon $P$ and an axis-aligned box $K$ containing $P$ in its interior, find the smallest cardinality set $\RR\cup\BB$ of axis-aligned boxes such that:
\begin{itemize}
    \item the boxes in $\RR$ cover $P$ exactly,
    \item and the boxes in $\BB$ cover $K\setminus P$ exactly.
\end{itemize}
\end{definition}

We start with the following:
\begin{theorem}\label{thm:reduction-SBCC-SRPC}
If SBCC can be approximated in polynomial time within a constant factor of $(1+\epsilon)$, then SRPC can also be approximated in polynomial time within a constant factor of $(1+\epsilon)$.
\end{theorem}
The reduction here is almost identical to the previously mentionned one from RPC to BCC. We include it for the sake of completeness.

\begin{figure}
    \centering
    \includegraphics[scale=1]{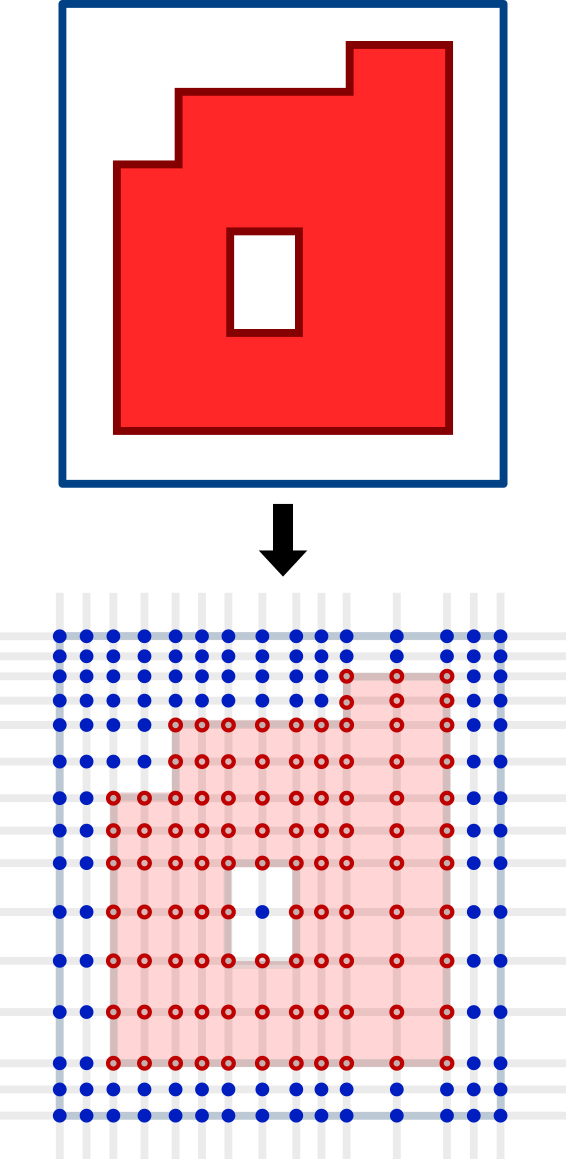}
    \caption{Reduction from SRPC to SBCC.}
    \label{fig:rectilinear_cover_reduction}
\end{figure}
\begin{proof}
We first show how to create a corresponding SBCC instance from a SRPC instance in polynomial time. Consider an axis aligned box $K$ and a closed rectilinear polygon $P$ contained in its interior (this constitutes our instance of the SRPC problem). Consider the set $L_1$ of all axis aligned lines which pass through a vertex of $P$ or $K$.
Between any two consecutive vertical (resp.\ horizontal) lines of $L_1$ draw a vertical (resp.\ horizontal) line. Call the set of newly drawn lines $L_2$. The SBCC instance we consider is the set $S=R\cup B$ of pairwise intersections between lines of $L_1\cup L_2$, colored red if they are in $P$ or on its boundary and blue otherwise. See Figure \ref{fig:rectilinear_cover_reduction} for an illustration. 

For any solution to the SRPC problem on $P$, we can slightly extend the boxes covering $P$ and shrink those covering the complement of $P$ to obtain a solution of the SBCC problem on $S$ of the same cardinality (in polynomial time). 

Let us see how we can go in the reverse direction. Consider some solution $\RR\cup\BB$ to the SBCC problem on $S$. Start by expanding the boxes in the solution in the four directions as much as possible without covering a point of the opposite color. Now for each red box $K_R$, replace it with the smallest axis-aligned box $K'_R$ containing $K_R\cap P$. For each blue box $K_B$, replace it the smallest axis-aligned box $K'_B$ containing $K_B\setminus P$. Notice that every point of $S$ which does not lie on the boundary of $P$ is still covered. Now consider some cell $c$ of the grid $L_1 \cup L_2$. By construction, one of the corner points $p$ of $c$ is the intersection of two lines in $L_1$. Any box covering this corner has been expanded to cover the whole cell (possibly excluding the edges of the cell contained in edges of $P$ with the). Thus $P$ is covered by red boxes and $K\setminus P$ is covered by blue boxes. It remains to show that these covers are exact. Suppose some blue box $K'_B$ intersects the interior of $P$. Because $K'_B$ covers at least one blue point $b$, its interior intersects $K\setminus P$ and thus also an edge $e$ of $P$. Because $K'_B$ covers no red point, no vertex of $P$ lies inside $K'_B$ and $e$ must join two opposite edges of $K'_B$. Thus there is an axis aligned line passing through $b$ and intersecting $e$ inside $K'_B$. By construction a red point lies on this intersection. This is a contradiction as $K'_B$ covers no red point. We conclude that $K'_B$ does not intersect the interior of $P$. The same reasoning shows that no red box $K'_R$ intersects $K\setminus P$.

In short, the optimal solutions for the SRPC problem on $P$ and the SBCC problem on $S$ have the same size, and any solution to the latter can be transformed into a solution to the former of the same size in polynomial time.
\end{proof}

Ideally at this point we would like to show that, say, the existence of a $\PTAS$ for SRPC implies the existence of a $\PTAS$ for RPC. Then, because a $\PTAS$ for RPC cannot exist unless $\P=\NP$, this would imply that the same holds for SRPC and thus also for SBCC. Unfortunately, the trivial reduction from RPC to SRPC is not approximation-preserving. Intuitively, an approximation to SRPC can be good overall because the optimal cover of the complement is large and well approximated while the optimal cover of the polygon itself is small and poorly approximated (thus yielding a poor solution to the RPC instance). To get around this, we focus on polygons where the size of the optimal cover of the complement is upper-bounded by a constant times the size of the optimal cover of the polygon itself.

\begin{definition}
Let $P$ be a closed rectilinear polygon and let $K$ be an-axis aligned box containing $P$ in its interior.
We say that $P$ has a \emph{$d$-small-complement} if the cardinality of the smallest exact box-cover of $K\setminus P$ is at most $d$ times the cardinality of the smallest exact box-cover of $P$.
\end{definition}

Note that the specific choice of $K$ here does not impact the definition.

\begin{lemma}
If SRPC can be approximated within a constant factor of $(1+\epsilon)$ in polynomial time, then RPC on rectilinear polygons with $d$-small-complements can be approximated within a constant factor of $(1+(d+1)\epsilon)$ in polynomial time.
\end{lemma}
\begin{proof}
Let $P$ be a closed polygon with a $d$-small-complement, and let $K$ be an-axis aligned box containing $P$ in its interior. Let $\opt$ (resp.\ $\opt_\mathrm{sim}$ denote the minimum cardinality of a RPC (resp.\ SRPC) on $P$. Let $\overline{\opt}$ denote the minimum exact box-cover of $K\setminus P$. We have $\opt_\mathrm{sim} = \opt + \overline{\opt}$ and $\overline{\opt} \leq d\cdot\opt$.
Consider some solution of size $\mathrm{sol_{sim}}\leq (1+\epsilon)\opt_\mathrm{sim}$ to the SRPC on $P$. This solution $\mathrm{sol_{sim}}$ consists of a solution to the RPC problem on $P$ of size $\mathrm{sol}$  together with an exact box-cover of $K\setminus P$ of size $\overline{\mathrm{sol}}$. Finally we have 
\begin{align*}
    \mathrm{sol} &= \mathrm{sol_{sim}} - \overline{\mathrm{sol}} \\
    &\leq (1+\epsilon)\opt_\mathrm{sim} - \overline{\mathrm{sol}}\\
    &\leq (1+\epsilon)\opt + (1+\epsilon)\overline{\opt} - \overline{\mathrm{sol}}\\
    &\leq (1+\epsilon)\opt + \epsilon\overline{\opt}\\
    &\leq (1+\epsilon)\opt + \epsilon d\cdot\opt\\
    &\leq (1+(d+1)\epsilon)\opt
\end{align*}
\end{proof}

We have the following criteria to identify polygons with $d$-small-complements
\begin{lemma}
Let $\alpha\geq 0$ be a constant and let $P$ be a closed rectilinear polygon with $c$ convex vertices and $h\leq\alpha\cdot c$ holes. Then $P$ has a $(4+8\alpha)$-small-complement. 
\end{lemma}
\begin{proof}
Let $K$ be an-axis aligned box containing $P$ in its interior. Let $r$ be the number of reflex vertices of $P$ and $n=r+c$ be the number of vertices of $P$. Let $s$ (resp.\ $\overline{s}$) denote the cardinality of the smallest exact box-cover of $P$ (resp.\ $K\setminus P$).  By Lemma \ref{lemma:refl-conv-holes}, we have $r=c+4(h-1)\leq c(1+4\alpha)-4$.
Because any box in an exact cover can cover at most 4 convex vertices of $P$, we have $s \geq c/4$. The set $K\setminus P$ is a disjoint union of rectilinear polygons with a total of $n+4$ vertices, and only one of these polygons has a hole (corresponding to $P$). By Lemma \ref{lemma:ind-rec-cover}, $K\setminus P$ can be exactly covered with at most $(n+4)/2 = (r+c+4)/2 \leq \frac{c}{4}(4+8\alpha)$ boxes. Thus $\overline{s} \leq \frac{c}{4}(4+8\alpha) \leq s(4+8\alpha)$ and the claim holds.
\end{proof}

If we chain all these implications together, we get the following:
\begin{theorem}
If there is a $\PTAS$ for SBCC, then there is a $\PTAS$ for RPC restricted to rectilinear polygons where the number of holes is at most $\alpha$ times the number of convex vertices, for any constant $\alpha \geq 0$.
\end{theorem}

Berman and DasGupta showed that there is no $\PTAS$ for RPC unless $\P=\NP$ \cite{Berman1997}, by reducing Vertex Cover to RPC in an approximation-preserving way. One interesting (and in our case, useful) thing to note is that the instances of RPC produced by this reduction are rectilinear polygons where every hole contributes at least one convex vertex to the polygon. In particular these instances have at least as many convex vertices as they have holes. Thus, we can state Berman and DasGupta's result in a slightly more precise way as follows:
\begin{theorem}[\cite{Berman1997}]
There is no $\PTAS$ for RPC unless $\P=\NP$, even when restricted to rectilinear polygons with no more holes than convex vertices.
\end{theorem}

This combined with the previous theorem now immediately yields:

\begin{theorem}
There is no $\PTAS$ for SBCC unless $\P=\NP$.
\end{theorem}

Note that the only place where we have needed to exploit sets of points which are not in general position is in the proof of Theorem \ref{thm:reduction-SBCC-SRPC}. 

\section{A constant-factor approximation for Simultaneous Boxes Class-Cover}

Here we prove that SBCC can be approximated within a constant factor in polynomial time. To do so, we show that the minimum size of a SBCC of $S$ is bounded above and below by the minimum size of a SBCC of $S$ with interior-disjoint boxes. This latter problem can be approximated within a constant factor using the results from \cite{Mitchell1993}.

\begin{figure}
    \centering
    \includegraphics[scale=0.75]{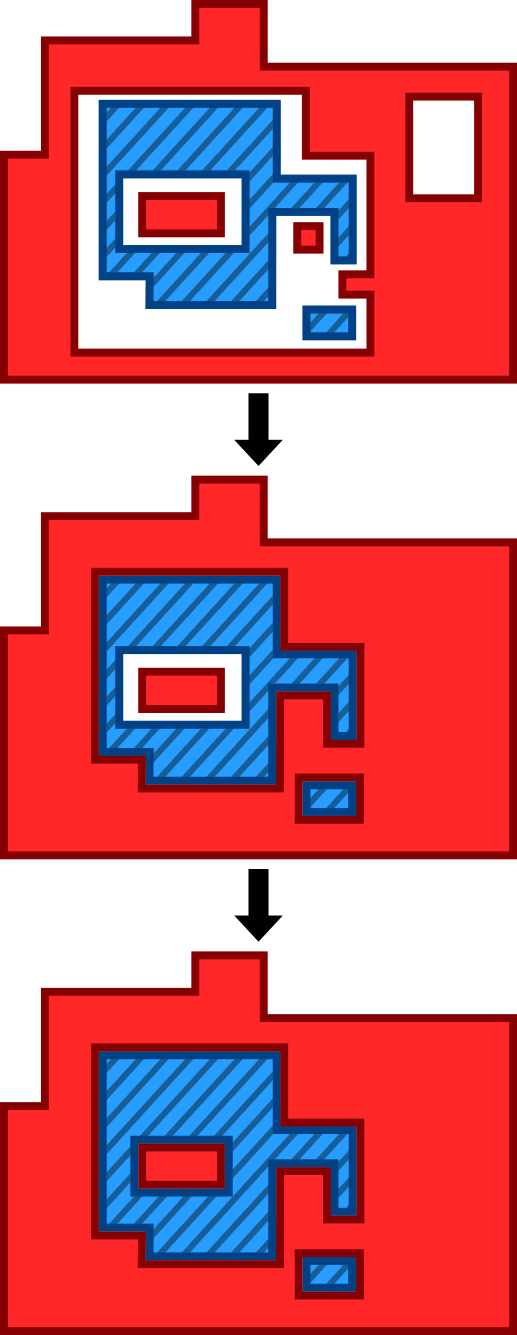}
    \caption{Illustration of the $\FILL$ procedure in the proof of Theorem \ref{thm:independent-vs-normal-cover}, first applied to the holes of the larger red region then to the hole of the larger blue region.}
    \label{fig:one_step}
\end{figure}

\begin{theorem}\label{thm:independent-vs-normal-cover}
The size of the minimum SBCC of $S$ with interior-disjoint boxes is at most $9$ times larger than the minimum SBCC of $S$.
\end{theorem}
\begin{proof}
Let $(\RR, \BB)$ be an SBCC of $S=R\cup B$ of size $k$. By shrinking the boxes in $\RR$ (resp.\ $\BB$), we can assume without loss of generality that every pair of boxes in $\RR$ (resp.\ $\BB$) either intersects at more than one point or is disjoint, and that blue boxes are disjoint from red boxes. The union of all boxes thus determines a partition $\Pi$ of the plane into red rectilinear polygons (red regions), blue rectilinear polygons (blue regions) and uncovered regions (which also have rectilinear boundaries). Call a bounded uncovered region of $\Pi$ a hole of $\Pi$ (notice the distinction between a hole of $\Pi$ and a hole of some region in $\Pi$). Call outer-complexity of $\Pi$ the sum of outer-complexities of all red and blue regions in $\Pi$.

Call $\FILL$ the process of taking a hole $H$ of $\Pi$ and giving it the color of the polygon $K$ that surrounds it. Suppose without loss of generality that $K$ is red. By applying the $\FILL$ procedure we have replaced the region $K$ of $\Pi$ with a new red rectilinear polygon $K'$ with the same outer-hull as $K$. Notice that this increases neither the number of red or blue regions in $\Pi$, and decreases the number of holes by one. Moreover, it does not increase the outer-complexity of $\Pi$, as any blue region or red region which is not adjacent to $H$ is unaffected, while any red region adjacent to $H$ gets merged with $K'$ and thus does not contribute to the outer-complexity of $\Pi$ any longer. The red and blue regions of $\Pi$ are still disjoint and still correctly cover all points of $S$ by construction.

Repeatedly apply $\FILL$ as long as there are holes in $\Pi$. This procedure terminates as $\Pi$ starts out with a finite number of holes and every application of $\FILL$ decreases the number of holes. Moreover, because $\Pi$ no longer has holes by the end of the procedure, every hole of every region in $\Pi$ must coincide with the outer-hull of some other region. On the other hand, the outer-hull of any region can coincide with at most one such hole. Recall that $k$ is the number of boxes we started with in the original SBCC. Because $\Pi$ starts out with at most $k$ colored regions and $\FILL$ never increases the number of colored regions, $\Pi$ still has at most $k$ colored regions.


By Lemma \ref{lemma:ind-rec-cover}, any colored region $P$ of $\Pi$ can be exactly covered with a number of boxes which is at most 
\begin{align*}
    &|P|/2 + |\holes(P)| -1 \\
    &=|\oh(P)|/2+\mysum{H \in \holes(P)}{}{(|H|/2+1)}-1.
\end{align*}

Because every colored region in $\Pi$ coincides with at most one hole of another colored region, and every hole in every colored region coincides with another colored region (more precisely, its outer-hull) by summing over all colored regions $P$ we get a total number of boxes smaller than
\begin{align*}
    &\sum_{P}{|\oh(P)|/2} + \sum_{P}{}{(|\oh(P)|/2+1)}\\
    \leq& \sum_{P}{(|\oh(P)|+1)}.
\end{align*}
Because $\Pi$ contains at most $k$ colored regions, this is at most $c+k$, where $c$ is the outer-complexity of $\Pi$. By Lemma \ref{lemma:refl-conv-holes} and because the applications of $\FILL$ did not increase the outer-complexity of $\Pi$, we have $c\leq 8k$. Thus the claim holds.
\end{proof}

In \cite{Mitchell1993}, Mitchell showed the following:
\begin{theorem}[\cite{Mitchell1993}]\label{thm:independent-cover}
Given a set $S$ of $n$ points and a set of axis-aligned boxes $\RR$, we can, in polynomial time, find a $O(1)$-approximation for the minimum cardinality set of interior-disjoint axis-aligned sub-boxes of $\RR$ which cover all points in $S$.
\end{theorem}

Using this result, we get our main theorem in this section:
\begin{theorem}
There is a polynomial-time algorithm to approximate minimum simultaneous class cover within a constant factor.
\end{theorem}
\begin{proof}
Consider some set of $n$ red and blue points $S$. Call $\opt$ the minimum  size of a SBCC of $S$ and $\opt_\mathrm{ind}$ the minimum size of a SBCC of $S$ with interior-disjoint boxes. 

Let $\RR$ be the set of monochromatic boxes on $S$. We shrink these boxes appropriately to consider only boxes which have a point of $S$ near every edge, so that the number of boxes in $\RR$ is polynomial in $n$ (and $\RR$ can also be computed in $O(\mathrm{poly}(n))$ time). Using Theorem \ref{thm:independent-cover} with this set $\RR$, we can get a SBCC with interior-disjoint boxes (which is also a SBCC by definition) of size $k \in O(\opt_\mathrm{ind})$ in $O(\mathrm{poly}(n))$ time.
By Theorem \ref{thm:class-cover-reduction} we have $\opt_\mathrm{ind} \leq 9\cdot \opt$. Thus $k \in O(\opt)$ and the claim holds.
\end{proof}

\section{Conclusion}
In this paper we have given an alternative proof of the $\APX$-hardness of the Boxes Class Cover problem which yields a more precise statement on the lower bound for a polynomial-time approximation factor (under the assumption that $\P\neq \NP$), and works as well for the case where half-strips are used instead of boxes. We have also explored the related Simultaneous Boxes Class Cover problem, giving proofs for a lower and upper bound which match up to a constant factor. These exhibit interesting connections with the Independent Sub-Box Cover problem studied by Mitchell \cite{Mitchell1993}. The problem of finding a constant-factor approximation to the original BCC problem remains open. Perhaps the methods used here could help find better upper bounds for this problem, or match the existing upper-bounds by more elementary means (as current upper bounds rely on the computation of weighted $\epsilon$-nets).


\small
\bibliographystyle{abbrv}
\bibliography{main}




\end{document}